\documentclass[10pt, conference, letterpaper]{IEEEtran}
\IEEEoverridecommandlockouts
\usepackage{cite}
\usepackage{amsmath,amssymb,amsfonts}
\usepackage{algorithmicx}
\usepackage[ruled, linesnumbered, vlined]{algorithm2e}
\usepackage{array}
\usepackage{graphicx}
\usepackage{stfloats}
\usepackage{amsthm}
\usepackage{subfigure}
\usepackage{url}
\usepackage{verbatim}
\usepackage{algpseudocode}  
\usepackage{textcomp}
\usepackage{xcolor}
\usepackage{color}
\newenvironment{sequation}{\begin{equation}\small}{\end{equation}}

\newtheorem{lemma}{\textbf{Lemma}}
\newtheorem{theorem}{\textbf{Theorem}}

\newtheorem{definition}{\textbf{Definition}}
\definecolor{r}{rgb}{1, 0, 0}
\definecolor{b}{rgb}{0, 0, 0}
\def\BibTeX{{\rm B\kern-.05em{\sc i\kern-.025em b}\kern-.08em
    T\kern-.1667em\lower.7ex\hbox{E}\kern-.125emX}}
\begin{document}

\title{An Online Joint Optimization Approach for QoE Maximization in UAV-Enabled Mobile Edge Computing\\
%{\footnotesize \textsuperscript{*}Note: Sub-titles are not %captured in Xplore and
%should not be used}

}

\author{\IEEEauthorblockN{Long He\IEEEauthorrefmark{2}, Geng Sun\IEEEauthorrefmark{2}$^*$, Zemin Sun\IEEEauthorrefmark{2}, Pengfei Wang\IEEEauthorrefmark{3}, Jiahui Li\IEEEauthorrefmark{2}, Shuang Liang\IEEEauthorrefmark{4}, Dusit Niyato\IEEEauthorrefmark{5}}
	
\IEEEauthorblockA{\IEEEauthorrefmark{2}{College of Computer Science and Technology, Jilin University, Changchun 130012, China} \\
\IEEEauthorrefmark{3}{School of Computer Science and Technology, Dalian University of Technology, Dalian 116024, China}\\
\IEEEauthorrefmark{4}{School of Information Science and Technology, Northeast Normal University, Changchun 130024, China}\\
\IEEEauthorrefmark{5}{School of Computer Science and Engineering, Nanyang Technological University, Singapore 639798, Singapore}\\
E-mails:\{helong0517, lijiahui0803\}@foxmail.com, \{sungeng, sunzemin\}@jlu.edu.cn, \\
wangpf@dlut.edu.cn, liangshuang@nenu.edu.cn, dniyato@ntu.edu.sg
}
\IEEEauthorrefmark{1}{Corresponding author: Geng Sun }\\
}	

\maketitle

\begin{abstract}
\par Given flexible mobility, rapid deployment, and low cost, unmanned aerial vehicle (UAV)-enabled mobile edge computing (MEC) shows great potential to compensate for the lack of terrestrial edge computing coverage. However, limited battery capacity, computing and spectrum resources also pose serious challenges for UAV-enabled MEC, which shorten the service time of UAVs and degrade the quality of experience (QoE) of user devices (UDs) {\color{b} without effective control approach}. In this work, we consider a UAV-enabled MEC scenario where a UAV serves as an aerial edge server to provide computing services for multiple ground UDs. Then, a joint task offloading, resource allocation, and UAV trajectory planning optimization problem (JTRTOP) is formulated to maximize the QoE of UDs under the UAV energy consumption constraint. To solve the JTRTOP that is proved to be a future-dependent and NP-hard problem, an online joint optimization approach (OJOA) is proposed. Specifically, the JTRTOP is first transformed into a per-slot real-time optimization problem (PROP) by using the Lyapunov optimization framework. Then, a two-stage optimization method based on game theory and convex optimization is proposed to solve the PROP. Simulation results validate that the proposed approach can achieve superior system performance compared to the other benchmark schemes.
\end{abstract}

%\begin{IEEEkeywords}
%Unmanned aerial vehicle, mobile edge computing, convex optimization, game theory, Lyapunov optimization. 
%\end{IEEEkeywords}
%a
%Introduction
%
\section{Introduction}
\label{sec:Introduction}
\par \IEEEPARstart{W}{ith} artificial intelligence and wireless communications development, many intelligent applications with strict requirements on computing resources and latency have emerged explosively\cite{DongSQWZWW21}, such as real-time video analysis~\cite{hou2023eavs}, virtual reality/augmented reality~\cite{aug_reality}, and interactive online games~\cite{shi2016edge}. However, the limited battery capacity and computing capability of user devices (UDs) make it difficult to maintain a high-level quality of experience (QoE) for these intelligent applications~\cite{Hekmati}. To overcome this challenge, mobile edge computing (MEC) has emerged as a promising paradigm to offer cloud computing resources in close proximity to UDs~\cite{Mao2017,QuDWWWTTD22}. Specifically, UDs can offload latency-sensitive and computation-hungry tasks to edge servers to improve the QoE. Equipped with cloud computing capabilities, the edge servers can concurrently provide real-time and energy-efficient computing services for multiple UDs. However, conventional terrestrial MEC still faces the challenges of limited network coverage and high deployment cost due to the dependence on ground infrastructures, especially in remote areas~\cite{Mozaffari2019}.
\par The limitations of conventional terrestrial MEC have prompted a paradigm shift toward UAV-enabled MEC due to the line-of-sight (LoS) communication, high maneuverability, and flexible deployment of UAVs~\cite{Li2023TMC,li2023multi,qu2023elastic}. First, the high probability LoS links of UAVs boost the communication coverage, network capacity, and reliable connectivity~\cite{ApostolopoulosF23,VamvakasTP19b}. Furthermore, their flexible mobility enables rapid and on-demand deployment, especially in distant areas where terrestrial infrastructures are unavailable. Besides, the integration of UAV and MEC offers flexible computing capabilities to improve the QoE of UDs.

\par However, several fundamental challenges should be overcome to fully exploit the benefits of UAV-enabled MEC. \textbf{\textit{i) Resource Allocation.}} Various tasks of UDs are generally heterogeneous and time-varying, and they have stringent requirements for the offloading service. However, the limited computing resources and scarce spectrum resources of UAV-enabled MEC and the stringent demands of UDs could lead to the competition for resources inside the MEC server, especially during peak times. Thus, under resource constraints, it is challenging for the MEC server to determine an efficient resource allocation strategy to meet the demands of various tasks. \textbf{\textit{ii) Task Offloading.}} The offloading decision of each UD depends not only on its own offloading demand but also on the offloading decisions of the other UDs, which makes the offloading decisions among UDs coupling and complex. \textbf{\textit{iii) Trajectory Planning.}} Although the mobility of UAVs increases the flexibility and elasticity of MEC, it also brings significant difficulties in UAV trajectory planning. \textbf{\textit{iv) Energy Constraint.}} The limited onboard battery capacity of UAVs leads to finite service time, which makes it challenging to balance the service time of UAVs and the QoE of UDs. In addition, under the constraints of UAV's resources and energy, the resource allocation strategy of UAVs, the task offloading decisions of UDs, and the trajectory planning of UAVs have mutual effects on each other, leading to the complexity of the decision-making process. 
\par {\color{b}  To overcome the aforementioned challenges, we propose an online approach for joint optimization of task offloading, resource allocation, and UAV trajectory planning to maximize the QoE of UDs under the UAV energy consumption constraint.} The main contributions are summarized as follows:
\begin{itemize}
\item \textbf{\textit{System Architecture.}} We consider a stochastic UAV-enabled MEC system with energy and resource constraints consisting of a UAV and multiple ground UDs. Specifically, the UAV is employed as an aerial edge server relying on limited battery capacity, computing and communication resources to provide computing services to UDs with time-varying computation requirements and dynamic mobility.

\item \textbf{\textit{Problem Formulation.}} We formulate a novel joint task offloading, resource allocation, and UAV trajectory planning optimization problem (JTRTOP) with the aim of maximizing the QoE of UDs under the UAV energy consumption constraint. Specifically, the QoE of UDs is theoretically measured by synthesizing the completion delay of the tasks and energy consumption of UDs.
	
\item \textbf{\textit{Algorithm Design.}} Since the JTRTOP not only requires future information but is also non-convex and NP-hard, we propose an online joint optimization approach (OJOA) to solve the problem. Specifically, we first transform the JTRTOP into a per-slot real-time optimization problem (PROP) by using the Lyapunov optimization framework. Then, we propose a two-stage method to optimize the task offloading, resource allocation, and UAV position of PROP by using convex optimization and game theory.

\item \textbf{\textit{Validation.}} Both theoretical analysis and simulation experiments are performed to verify the effectiveness and performance of the proposed OJOA. Specifically, theoretical analysis demonstrates that the OJOA not only satisfies the UAV energy consumption constraint but also converges to a sub-optimal solution in polynomial time. Moreover, simulation results indicate that the proposed OJOA outperforms other benchmark schemes.
	
\end{itemize}
\par The remainder of the work is organized as follows. Section \ref{sec:Related Work} summarizes the related work. Section \ref{sec:System Model and problem Formulation} details the relevant system models and problem formulation. Section \ref{sec:Lyapunov-Based problem Transformation} describes the Lyapunov-based problem transformation. Section \ref{sec:Two-Stage Optimization Algorithm} presents the two-stage optimization algorithm and theoretical analysis. In Section \ref{sec:Simulation Results and Analysis}, simulation results are displayed and analyzed. Finally, Section \ref{sec:Conclusion} concludes the overall paper.
%
%Related works
%
\section{Related Work}
\label{sec:Related Work}
\par Most existing studies on UAV-enabled MEC are devoted to the design of offline algorithms to plan the entire task offload, resource allocation, and UAV trajectory, which assume that the locations of UDs are invariant and the computing requirements of UDs are fixed or known in advance~\cite{Xu2021,UAV-H,Hu2019}. However, many edge computing scenarios change dynamically over time, such as real-time video analysis and interactive online games, which means that the computing tasks arrive stochastically, the computing requirements of UDs are time-varying and the UDs are dynamically mobile. Therefore, it is necessary to design real-time decision-making algorithms without future information. 

\par There are also some works studying real-time decision-making. For example, Yang et al.~\cite{Yang2022} studied the UAV-enabled MEC system with random task arrival and user mobility. Specifically, the UAV trajectory and resource allocation were decided in real time to minimize the average energy consumption of all users through online algorithms based on Lyapunov optimization. Considering the time-varying computing requirements of user equipment, Wang et al.~\cite{Wang2022} jointly optimized the user association, resource allocation and trajectory of UAVs with the aim of minimizing energy consumption of all user equipment. To minimize the average power consumption of the system with randomly arriving user tasks, Hoang et al.~\cite{10102429} developed a Lyapunov-guided deep reinforcement learning framework. Zhou et al.~\cite{zhou2022two} proposed an alternating optimization-based algorithm by leveraging the Lyapunov optimization approach and dependent rounding technique to minimize the service delay.

{\color{b} \par In practice, due to the limited energy and computing resources of UDs, task completion delay and energy consumption are important indicators to measure the QoE of UDs. However, the abovementioned works mainly focus on minimizing the task completion delay and energy consumption of users (or the whole system) separately, which could not provide a high-level QoE for users. Furthermore, these works consider resource allocation from either the communication or the computation aspects, which may lead to severe performance degradation in practical UAV-enabled MEC systems where both communication and computing resources are insufficient. Motivated by these issues, in this work, we consider a stochastic UAV-enabled MEC system with time-varying computation requirements and dynamic mobility of UDs to minimize the user energy consumption and task completion latency simultaneously. Furthermore, the computing and communication resource allocation are jointly optimized.}
%
%System model and problem formulation
%
\section{System Model and Problem Formulation}
\label{sec:System Model and problem Formulation}
\begin{figure}[!hbt]
    \centering
    \includegraphics[width =3in]{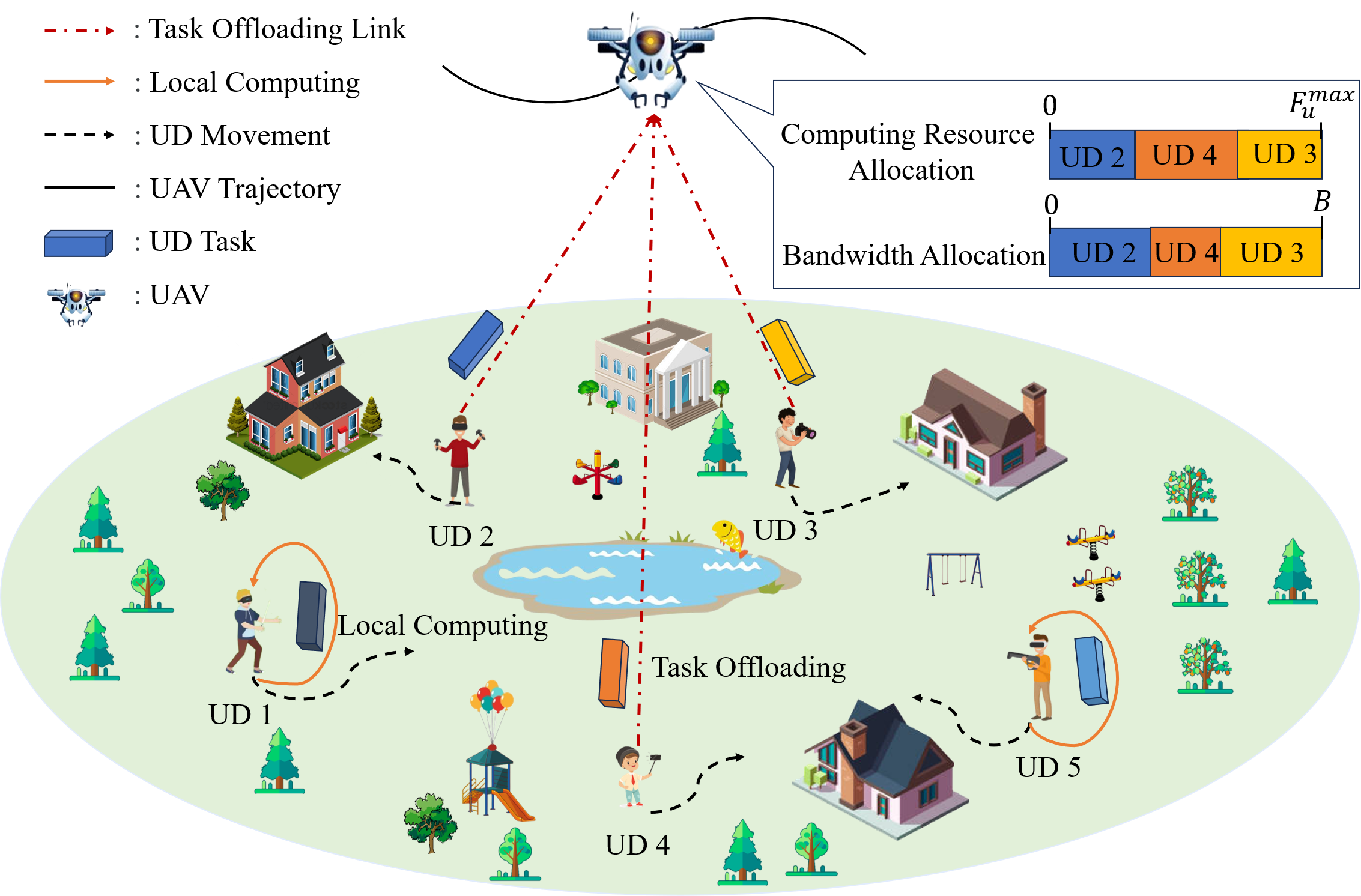}
    \caption{The UAV-enabled MEC consists a UAV and multiple ground UDs. The UAV provides computing services to UDs by allocating communication and computing resources. Each UD independently decides to compute its task locally or offload the task to the UAV.}
    \label{fig_gameModel}
\end{figure}
\par As illustrated in Fig. \ref{fig_gameModel}, we considered a UAV-enabled MEC system that consists of a rotary-wing UAV $u$ and $M$ UDs with the set $\mathcal{M}=\{1,2,\dots,M\}$. Equipped with MEC capability, the UAV is employed as an aerial edge server relying on limited battery capacity to provide computing offloading services to the UDs within a finite system timeline. Moreover, we discretize the system timeline into equal $T$ time slots~\cite{qu2021service}, i.e., $t\in \mathcal{T}=\{1,2,\dots,T\}$, wherein each slot duration is denoted as $\tau$. 
%{\color{r}The duration of time slot matches the timescale at which task offloading and resource allocation could be updated~\cite{wang2020joint}. The duration is chosen to be sufficiently small such that the locations of the UAV and the UDs are considered unchanged within each time slot regardless of the velocity. In addition, due to the mobility of ground users and the stochastic arrival of tasks, we also consider the UAV trajectory planning to improve the QoS.}
%
%Basic Model
%
\subsection{Basic Model}
\label{subsec:Basic Model}
\par \textbf{\textit{UD Model.}} We assume that each UD generates one computing task per time slot~\cite{Wang2022,JiangDXI23}. For UD $m\in\mathcal{M}$, the UD's attributes at time slot $t$ can be characterized as $\mathbf{St}_m^{\text{UD}}(t)=\left(f_m^{\text{UD}},\mathbf{\Phi}_m(t),\mathbf{P}_m(t)\right)$, where $f_m^{\text{UD}}$ denotes the local computing capability of UD $m$. The computing task generated by UD $m$ is characterized as $\mathbf{\Phi}_m(t)=\{D_m(t),\eta_m(t),T^{\text{max}}_m(t)\}$ at time slot $t$, wherein $D_m(t)$ represents the input data size (in bits), $\eta_m(t)$ denotes the computation intensity (in cycles/bit), and $T^{\text{max}}_m(t)$ is the maximum tolerable delay. $\mathbf{P}_m(t)=\left[x_m(t),y_m(t)\right]$ represents the location coordinates of UD $m$ at time slot $t$. Similar to~\cite{Tabassum2019,LiuSSLDS20}, the mobility of UDs is modeled as a Gauss-Markov mobility model, which is widely employed in cellular communication networks~\cite{BatabyalB15}. Specifically, the velocity of UD $m$ at time slot $t+1$ are updated as follows:
\begin{sequation}
\begin{split}
    \mathbf{v}_m(t+1)=\alpha \mathbf{v}_m(t)+(1-\alpha)\overline{\mathbf{v}}_m+\sqrt{1-\alpha^2}\mathbf{w}_m(t),
    \label{eq.vehcle_velocity}
\end{split}
\end{sequation}

\noindent where $\mathbf{v}_m(t)=(v_m^x(t),v_m^y(t))$ denotes the velocity vector at time slot $t$. $\alpha$ represents the memory level, which reflects the temporal-dependent degree and $\overline{\mathbf{v}}_m$ is the asymptotic means of velocity. $\mathbf{w}_m(t)$ is the uncorrelated random Gaussian process $N(0,\sigma_m^2)$, where $\sigma_m$ denotes the asymptotic standard deviation of velocity. Therefore, the mobility of UD $m$ can be updated as follows:
\begin{sequation}
\mathbf{P}_m(t+1) = \mathbf{P}_m(t) + \mathbf{v}_m(t)\tau.
\end{sequation}
\par \textbf{\textit{UAV Model.}} UAV $u$ is characterized by $\mathbf{St}^{u}(t)=(\mathbf{P}_{u}(t),H,F_{u}^{\text{max}},B)$, wherein $\mathbf{P}_u(t) = [x_u(t),y_u(t)]$ and $H$ represent the horizontal coordinate and flight height of the UAV at time slot $t$, respectively. $F_{u}^{\text{max}}$ represents the total computing resources and $B$ denotes the total bandwidth resources. 
%Typically, the propulsion energy is influenced by the UAV's flying speed and acceleration. For the sake of clarity and manageable analysis in this study, we neglect the extra/less energy consumption resulting from UAV acceleration/deceleration. This assumption is reasonable in scenarios where the duration of UAV maneuvering constitutes only a small fraction of the total operation time. 

\par \textbf{\textit{Decision Variables.}} The following decisions need to be made jointly. \textit{i) Task Offloading Decision.} For task $\mathbf{\Phi}_m(t)$, we define a binary variable $a_m(t)$ to represent the offloading decision of UD $m$ at time slot $t$, where $a_m(t)=0$ indicates that the task is processed locally, and $a_m(t)=1$ indicates that the task is offloaded to the UAV for processing. \textit{ii) Resource Allocation Decision.} For the UAV, the resources allocated to task $\mathbf{\Phi}_m(t)$ are denoted as $\{F_m(t),w_m(t)\}$ at time slot $t$, where $F_m(t)$ is the amount of allocated computing resources and $w_m(t)$ is the proportion of allocated bandwidth resources. \textit{iii) UAV Trajectory Planning.} For the UAV, trajectory planning can be expressed as a sequence of optimal positions for each time slot, i.e., $\mathbf{P}_u=\{\mathbf{P}_u(t)\}_{t\in \mathcal{T}}$.
\vspace{-3pt}
%
%Communication Model
%
\subsection{Communication Model}
\label{subsec:Communication Model}

\par The probabilistic line-of-sight (LoS) channel model is employed to model the communication between the UAV and UDs \cite{communicationmodel2}. First, the LoS probability $P_{m,u}^{\text{LoS}}(t)$ between UD $m$ and the UAV at time slot $t$ can be defined as~\cite{sun2023uav}
\begin{sequation}
    P_{m,u}^{\text{LoS}}(t)=\frac{1}{1+\xi_1 \exp(-\xi_2(\theta_{m,u}(t) - \xi_1))},\label{eq.P_LoS}
\end{sequation}

\noindent where $\xi_1$ and $\xi_2$ are constants depending on the propagation environment, $\theta_{m,u}(t) =\frac{180}{\pi} \arcsin{\frac{H}{d_{m,u}(t)}}$ denotes the elevation angle and $d_{m,u}(t)$ represents the straight-line distance between UD $m$ and the UAV. Similar to~\cite{Yang2022,communicationmodel1}, the channel power gain can be calculated as
\begin{sequation}
	\begin{split}
		g_{m,u}(t)&= P_{m,u}^{\text{LoS}}(t)\beta_0d_{m,u}^{-\tilde{\mu}}(t)+(1-P_{m,u}^{\text{LoS}}(t))\kappa \beta_0d_{m,u}^{-\tilde{\mu}}(t)\\
		&=\Tilde{P}_{m,u}^{\text{LoS}}(t)\beta_0d_{m,u}^{-\tilde{\mu}}(t),\label{eq}\\
	\end{split}
\end{sequation}

\noindent where $\Tilde{P}_{m,u}^{\text{LoS}}(t)\triangleq P_{m,u}^{\text{LoS}}(t)+(1-P_{m,u}^{\text{LoS}}(t))\kappa$, $\kappa$ is the additional attenuation factor, $\beta_0$ denotes the channel gain at the reference distance 1 m, and $\tilde{\mu}$ is the path loss exponent. Therefore, the spectral efficiency of UD $m$ can be expressed as
\begin{sequation}
\label{eq.transratio1}
    r_{m,u}(t)=\log_2\left(1+\frac{\phi_m(t)}{(||\mathbf{P}_u(t)-\mathbf{P}_m(t)||^2+H^2)^{\mu}}\right),
\end{sequation}

\noindent where $\phi_m(t)=\frac{P_m\beta_0\Tilde{P}_{m,u}^{\text{LoS}}(t)}{N_0}$, $\mu=\frac{\tilde{\mu}}{2}$, $P_{m}$ is the transmission power of UD $m$, and $N_0$ represents the noise power. 

\par Moreover, the widely used orthogonal frequency-division multiple access (OFDMA) is employed in the communication models. Therefore, the communication rate of UD $m$ at time slot $t$ can be presented as~\cite{Ndikumana}
\begin{sequation}
\label{eq.rate}
R_{m,u}(t) = w_m(t)Br_{m,u}(t),
\end{sequation}

%
%Computation Model
%
\subsection{Computation Model}
\label{subsec:Computation Model}
\par For task $\mathbf{\Phi}_m(t)$ generated by UD $m$, the task can be processed either locally on the UD or remotely on the UAV, which is determined by the UD's offloading decision $a_m(t)$.

\par \textbf{\textit{Local Computing.}} UD $m$ processes task $\mathbf{\Phi}_m(t)$ locally (i.e., $a_m(t)=0$). The local completion latency of the task at time slot $t$ can be calculated as
\begin{sequation}
    \label{eq.loc-delay}
    T_m^{\text{loc}}(t)=\frac{\eta_m(t)D_m(t)}{f_m^{\text{UD}}},
\end{sequation}

\par Accordingly, the energy consumption of UD $m$ to execute task $\mathbf{\Phi}_m(t)$ locally at time slot $t$ is calculated as~\cite{UAV-H}
\begin{sequation}
    \label{eq.loc-delay}
    E_m^{\text{loc}}(t)=k(f_m^{\text{UD}})^{3}T_m^{\text{loc}}(t),
\end{sequation}

\noindent where $k$ denotes the effective switched capacitance cofficient that depends on the hardware architecture of the UD.

\par \textbf{\textit{Edge Computing.}} Task $\mathbf{\Phi}_m(t)$ is offloaded to the UAV for processing (i.e., $a_m(t)=1$). In this case, the UAV allocates computing and communication resources to perform the task. The edge processing delay includes transmission delay and edge execution delay, which can be calculated as
%Based on the allocated resources, the task are execute on the UAV and the computing results are returned to the UD. Note that the delay of result feedback is omitted in the calculation of the service delay since the size of the computation results is much smaller than that of the input data for most of mobile applications.
\begin{sequation}
    \label{eq.ec-delay}
    T_m^{\text{ec}}(t) = \frac{D_m(t)}{R_{m,u}(t)}+\frac{\eta_m(t) D_m(t)}{F_m(t)}.
\end{sequation}

\par The energy consumption generated by processing the task at time slot $t$ consists of the transmission energy consumption of UD $m$ and the computation energy consumption of the UAV. The transmission energy consumption of UD $m$ at time slot $t$ can be calculated as
\begin{sequation}
    \label{eq.trans-energy}
    E_m^{\text{ec}}(t) = P_m\frac{D_m(t)}{R_{m,u}(t)}.
\end{sequation}

\par Then, the computation energy consumption of the UAV to execute task $\mathbf{\Phi}_m(t)$ can be given as~\cite{JiangDXI23}
\begin{sequation}
    \label{eq.uav-comp-energy}
    E_{m,u}^{\text{c}}(t) = \varpi\eta_m(t)D_m(t).
\end{sequation}

\noindent where $\varpi$ represents the UAV energy consumption per unit CPU cycle. Therefore, the total computation energy consumption of the UAV at time slot $t$ can be given as
\begin{sequation}
    \label{eq.comp-energy}
    E_u^{\text{c}}(t) = \sum_{m\in \mathcal{M}}a_m(t)E_{m,u}^{\text{c}}(t).
\end{sequation}

%
%Cost Model
%
\subsection{Cost Model}
\label{subsec:Cost Model}

\par\textbf{\textit{UD Cost.}} In this work, we consider that each UD's cost at time slot $t$ consists of the task completion delay and the UD's energy consumption, which reflects the UD's QoE. The completion delay of task $\mathbf{\Phi}_m(t)$ can be presented as
\begin{sequation}
    \label{eq.delay}
    T_m(t)=(1-a_m(t))T_m^{\text{loc}}(t)+a_m(t)T_m^{\text{ec}}(t).
\end{sequation}

\noindent Then, the energy consumption of UD $m$ can be given as
\begin{sequation}
    \label{eq.energy}
    E_m(t)=(1-a_m(t))E_m^{\text{loc}}(t)+a_m(t)E_m^{\text{ec}}(t).
\end{sequation}

\noindent Similar to~\cite{Chen2022,Ding2022}, the cost of UD $m$ at time slot $t$ can be formulated as
\begin{sequation}
    \label{eq.UD-cost}
    C_m(t)=\gamma_m T_m(t) + (1-\gamma_m)E_m(t),
\end{sequation}

\noindent where $\gamma_m$ and $1-\gamma_m$ represent the weighted parameters of delay and energy consumption of UD $m$ respectively, which can be flexibly set based on the UD's preference for delay and energy consumption. Obviously, minimizing the cost of UDs is equivalent to maximizing the QoE of UDs.

\par \textbf{\textit{UAV Energy Cost.}} Here, the cost of the UAV at time slot $t$ is expressed as the energy consumption, which includes the computing energy consumption and propulsion energy consumption. Similar to~\cite{communicationmodel1, pan2023joint}, the propulsion power consumption for a rotary-wing UAV with speed $v_u$ can be expressed as
\begin{sequation}
\label{eq.prop-energy}
P_u(v_u)=\underbrace{C_1\left(1+\frac{3 v_u^2}{U_{\text{p}}^2}\right)}_{\text {blade profile }}+\underbrace{C_2 \sqrt{\sqrt{C_3+\frac{v_u^4}{4}}-\frac{v_u^2}{2}}}_{\text{induced }}+\underbrace{C_4 v_u^3}_{\text {parasite}},
\end{sequation}

\noindent where $U_{\text{p}}$ refers to the rotor's tip speed, and $C1$, $C2$, $C3$, and $C4$ are constants described in~\cite{Yang2022}. Therefore, the energy consumption of the UAV at time slot $t$ can be given as
\begin{sequation}
    \label{eq.UAV-energy}
    E_u(t) = E_u^{\text{c}}(t) + E_u^{\text{p}}(t).
\end{sequation}

\noindent where $E_u^{\text{p}}(t)=P_u(v_u(t))\tau$ denotes the propulsion energy consumption at time slot $t$. To guarantee service time, we define the UAV energy consumption constraint as follows:
\begin{sequation}
    \label{eq.eng-cons}
    \lim _{T \rightarrow+\infty} \frac{1}{T} \sum_{t=1}^{T} \mathbb{E}\left\{E_u(t)\right\} \leq \bar{E}_u,
\end{sequation}

\noindent where $\bar{E}_u$ is the energy budget of the UAV per time slot.

%
%problem formulation
%
\subsection{Problem Formulation}
\label{subsec:problem Formulation}

\par The objective of this work is to minimize the average costs of all UDs over time (i.e., time-average UD cost), by jointly optimizing the task offloading strategy $\mathbf{A}=\{\mathcal{A}^t|\mathcal{A}^t=\{a_m(t)\}_{m\in\mathcal{M}}\}_{t\in\mathcal{T}}$, computing resource allocation $\mathbf{F}=\{\mathcal{F}^t|\mathcal{F}^t=\{F_m(t)\}_{m\in\mathcal{M}}\}_{t\in\mathcal{T}}$, communication resource allocation $\mathbf{W}=\{\mathcal{W}^t|\mathcal{W}^t=\{w_m(t)\}_{m\in \mathcal{M}}\}_{t\in \mathcal{T}}$, and trajectory planning $\mathbf{P}_u=\{\mathbf{P}_u(t)\}_{t\in \mathcal{T}}$. Therefore, the problem can be formulated as follows:

\begin{small}
 \begin{align}
    \textbf{P}: \quad &\underset{\mathbf{A}, \mathbf{F}, \mathbf{W}, \mathbf{P}_u}{\text{min}} \frac{1}{T}\sum_{t=1}^{T}\sum_{m=1}^{M}C_m(t) \label{P}\\
    \text{s.t.}\ \ 
    &\lim _{T \rightarrow+\infty} \frac{1}{T} \sum_{t=1}^{T} \mathbb{E}\left\{E_u(t)\right\} \leq \bar{E}_u,  \tag{\ref{P}{\text{a}}} \label{Pa}\\
    &a_m(t)\in \{0,1\}, \forall m\in \mathcal{M}, t\in \mathcal{T}, \tag{\ref{P}{\text{b}}} \label{Pb}\\
    &a_m(t)T^{\text{ec}}_m(t)\leq T_m^{\text{max}}, \forall m\in \mathcal{M}, t\in \mathcal{T}, \tag{\ref{P}{\text{c}}} \label{Pc}\\
    &0\leq F_m(t) \leq F_u^{\text{max}}, \forall m\in \mathcal{M}, t\in \mathcal{T}, \tag{\ref{P}{\text{d}}} \label{Pd}\\
    &\sum_{m=1}^M a_m(t)F_m\leq F_u^{\text{max}}, \forall t\in \mathcal{T}, \tag{\ref{P}{\text{e}}} \label{Pe}\\
    &0\leq w_m(t) \leq 1, \forall m\in \mathcal{M}, \forall t\in \mathcal{T}, \tag{\ref{P}{\text{f}}} \label{Pf}\\
    &\sum_{m=1}^M a_m(t) w_m(t)\leq 1, \forall t\in \mathcal{T}, \tag{\ref{P}{\text{g}}} \label{Pg}\\
    &\mathbf{P}_u(1)=\mathbf{P}_I, \tag{\ref{P}{\text{h}}} \label{Ph}\\
    &\left\|\mathbf{p}_u(t+1)-\mathbf{p}_u(t)\right\| \leq v_u^{\text{max}} \tau, \forall t\in \mathcal{T},\tag{\ref{P}{\text{i}}} \label{Pi}   
\end{align}
\end{small}

\noindent Constraint (\ref{Pa}) is the long-term energy consumption constraint of the UAV. Constraint (\ref{Pb}) indicates that each UD can only select one strategy as its offloading decision. Constraint (\ref{Pc}) means that the completion delay of edge computing should not exceed the maximum tolerance delay. Constraints (\ref{Pd}) and (\ref{Pe}) imply that the allocated computing resources should be a positive value and not exceed the total amount of computing resources owned by the UAV. Constraints (\ref{Pf}) and (\ref{Pg}) limit the allocation of communication resources. Constraints (\ref{Ph})-(\ref{Pi}) are the constraints on trajectory planning.

\par \textbf{\textit{Challenges.}} {\color{b} There are two main challenges to obtain the optimal solution of problem $\textbf{P}$. \textit{i) Future-dependent.}}
Optimally solving problem $\textbf{P}$ requires complete future information, e.g., task computing demands and locations of all UDs across all time slots. However, obtaining the future information is very challenging in the considered time-varying scenario. {\color{b} \textit{ii) Non-convex and NP-hard.}} Problem $\textbf{P}$ contains both binary variables (i.e., task offloading decision $\mathbf{A}$) and continuous variables (i.e., resource allocation $\{\mathbf{F},\mathbf{W}\}$ and UAV’s trajectory $\mathbf{P}_u$) is an mixed-integer non-linear programming (MINLP) problem, which is non-convex and NP-hard~\cite{boyd2004convex,belotti2013mixed}. Therefore, solving the problem directly remains challenging even with knowledge of the future information.
%
%Lyapunov-Based problem Transformation
%
\section{Lyapunov-Based Problem Transformation}
\label{sec:Lyapunov-Based problem Transformation}
\par {\color{b} Since problem $\textbf{P}$ is future-dependent, an online approach is necessary to make real-time decisions without foreseeing the future. Lyapunov-based optimization framework is a commonly adopted method for designing online algorithms~\cite{Yang2022,DingWCLC23}, which has the advantage of being simple and effective. To this end, we first transform problem $\textbf{P}$ into a per-slot real-time optimization problem based on the Lyapunov optimization framework.}
\par Firstly, to satisfy the UAV energy constraint (\ref{Pa}), we define two virtual energy queues $Q_u^{\text{c}}(t)$ and $Q_u^{\text{p}}(t)$ to represent the computing energy queue and the propulsion energy queue at time slot $t$ based on Lyapunov optimization technique, respectively. We assume that the queues are set as zero at the initial time slot, i.e., $Q_u^{\text{c}}(1)=0$ and $Q_u^{\text{p}}(1)=0$. Therefore, the virtual energy queues can be updated as
\begin{sequation}
    \label{eq.queue}
    \begin{cases}
        Q_u^{\text{c}}(t+1)=\max \left\{Q_u^{\text{c}}(t)+E_u^{\text{c}}(t)-\bar{E_u^{\text{c}}}, 0\right\}, \forall t \in \mathcal{T},\\
        Q_u^{\text{p}}(t+1)=\max \left\{Q_u^{\text{p}}(t)+E_u^{\text{p}}(t)-\bar{E_u^{\text{p}}}, 0\right\}, \forall t \in \mathcal{T},
    \end{cases}
\end{sequation}

\noindent where $\bar{E_u^{\text{c}}}$ and $\bar{E_u^{\text{p}}}$ represent the computation and propulsion energy budgets per slot, respectively and $\bar{E_u^{\text{c}}}+\bar{E_u^{\text{p}}}=\bar{E_u}$. Secondly, we define the \textit{Lyapunov function} $L(\mathbf{Q}_u(t))$, which represents a scalar measure of the queue backlogs, i.e.,
\begin{sequation}
    \label{eq.ly-func}
    L(\mathbf{Q}_u(t)) = \frac{(Q_u^{\text{c}}(t))^2+(Q_u^{\text{p}}(t))^2}{2}.
\end{sequation}

\noindent where $\mathbf{Q}_u(t)=\{Q_u^{\text{c}}(t),Q_u^{\text{p}}(t)\}$ is the vector of current queue backlogs. Thirdly, we define the \textit{conditional Lyapunov drift} for time slot $t$ as:
\begin{sequation}
    \label{eq.cond-ly-drift}
    \Delta L(\mathbf{Q}_u(t)) \triangleq \mathbb{E}\{L(\mathbf{Q}_u(t+1))-L(\mathbf{Q}_u(t)) \mid \mathbf{Q}_u(t)\}.
\end{sequation}

\noindent Finally, similar to~\cite{2010Neely,JiangDXI23,Yang2022}, the \textit{drift-plus-penalty} can be given as
\begin{sequation}
    \label{eq.drift-plus-penalty}
   D(\mathbf{Q}_u(t)) =\Delta L(\mathbf{Q}_u(t))+V \mathbb{E}\left\{C_s(t)\mid \mathbf{Q}_u(t)\right\},
\end{sequation}

\noindent where $C_s(t)=\sum_{m=1}^{M}C_m(t)$ is the total cost of all UDs at time slot $t$, and $V$ is a parameter that trades off the total cost and queue stability. 
\begin{theorem}
\label{the:drift-plus-penalty}
For all $t$ and all possible queue backlogs $\mathbf{Q}_u(t)$, the drift-plus-penalty is upper bounded as
\begin{sequation}
    \label{eq.theorem1}
    \begin{aligned}
    D(\mathbf{Q}_u(t)) \leq & W + {Q}_u^{\mathrm{c}}(t)(E_u^{\mathrm{c}}(t)-\bar{E_u^{\mathrm{c}}})\\
    &+{Q}_u^{\mathrm{p}}(t)(E_u^{\mathrm{p}}(t)-\bar{E_u^{\mathrm{p}}})+V\times C_s(t),
    \end{aligned}
\end{sequation}

\noindent where $W=\frac{1}{2} \max \left\{\left(\bar{E_u^{\mathrm{c}}}\right)^2,\left(E_{\max }^{\mathrm{c}}-\bar{E_u^{\mathrm{c}}}\right)^2\right\}+\frac{1}{2} \max \left\{\left(\bar{E_u^{\mathrm{p}}}\right)^2,\left(E_{\max }^{\mathrm{p}}-\bar{E_u^{\mathrm{p}}}\right)^2\right\}$ is a finite constant.
\end{theorem}
\begin{proof}
The proof can refer to Theorem 1 in~\cite{Yang2022}. Due to the space limit, we omit the details.
\end{proof}
\par {\color{b} According to the Lyapunov optimization framework, we minimize the right-hand side of inequality (\ref{eq.theorem1}).} Therefore, problem $\textbf{P}$ that relies on future information is transformed into the real-time optimization problem $\textbf{P}^{'}$ solvable with only current information, which is given as follows:

{\small \begin{align}
\textbf{P}^{\prime}:\ &\underset{\mathcal{A}^t,\mathcal{F}^t,\mathcal{W}^t,\mathbf{P}_{u^\prime}}{\text{min}}Q_u^{\text{c}}(t)E_u^{\text{c}}(t)+Q_u^{\text{p}}(t)E_u^{\text{p}}(t)+V\sum_{m=1}^{M} C_m(t) \label{P_temp} \\
\text{s.t.} \ &(\ref{Pb})-(\ref{Pi}) \notag
\end{align}}

\noindent where $\mathbf{P}_{u^\prime}=\mathbf{P}_{u}(t+1)$ represents the UAV position at time slot $t+1$. {\color{b} However, problem $\textbf{P}^{\prime}$ is still an MINLP problem and the decision variables are coupled to each other. Therefore, a large amount of computational overhead caused by seeking the optimal solution for problem $\textbf{P}^{\prime}$ may not be suitable for real-time decision making. To this end, we design a two-stage optimization method that obtains a sub-optimal solution in polynomial time complexity.} Furthermore, similar to~\cite{Cui2023}, we drop the time index for variables for the convenience of the following description.
%
% Two-Stage Optimization Algorithm
%
\section{Two-Stage Optimization Algorithm}
\label{sec:Two-Stage Optimization Algorithm}

\par In the section, a two-stage optimization method is proposed to solve the transformed problem $\mathbf{P}^\prime$. In the first stage, assuming a feasible $\mathbf{P}_{u^\prime}$, we optimize the task offloading decision $\mathcal{A}$ and resource allocation $\{\mathcal{F},\mathcal{W}\}$. In the second stage, based on the obtained task offloading decision $\mathcal{A}^{*}$ and resource allocation $\{\mathcal{F}^{*},\mathcal{W}^{*}\}$, we optimize the UAV position $\mathbf{P}_{u^\prime}$.
%
% Task Offloading Decision
%
\subsection{Stage 1: Task Offloading and Resource Allocation}
\label{subsec:Task Offloading Decision}

\par Assuming a feasible $\mathbf{P}_{u^\prime}$ and removing irrelevant constant terms, $\textbf{P}^{'}$ can be transformed into a subproblem $\textbf{P1}$ to decide task offloading and resource allocation, which is given as

 {\small \begin{align}
    \textbf{P1}:\quad &V\cdot \underset{\mathcal{A},\mathcal{F},\mathcal{W}}{\text{min}}\left(\frac{Q_u^{\text{c}}}{V}E_u^{\text{c}}+\sum_{m=1}^{M}C_m\right) \label{P1}\\
    \text{s.t.}\ \ 
    &(\ref{Pb})-(\ref{Pg}) \notag
\end{align}}

\par Problem $\textbf{P1}$ is still an MINLP problem, and the decisions of task offloading and resource allocation are coupled with each other. Considering that the UAV is dominant in the considered UAV-enabled MEC system, we prioritize resource allocation strategies for the UAV. Then, based on the resource allocation strategy, we optimize the UDs' offloading decisions.
%
% Resource Allocation
%
\subsubsection{Resource Allocation}
\label{subsubsec:Resource Allocation}
\par Given an arbitrary task offloading decision profile $\mathcal{A}$ of the UDs, the UAV decides resource allocation strategies to minimize problem $\textbf{P1}$. Define $s_m=\frac{F_m}{F_u^{\text{max}}}$, the resource allocation problem can be formulated as

 {\small \begin{align}
    \textbf{P1.1}:\quad &\underset{\mathcal{S},\mathcal{W}}{\text{min}}\sum_{m\in\mathbf{M}_1}\left[\gamma_m(\frac{D_m}{w_mBr_{m,u}}+\frac{\eta_mD_m}{s_mF_u^{\text{max}}})\right.\notag \\
    &\left.+(1-\gamma_m)\frac{P_mD_m}{w_mBr_{m,u}}\right] \label{P1.1}\\
    \text{s.t.}\ \
    & s_m\geq 0, \forall m\in \mathbf{M}_1, \tag{\ref{P1.1}{a}} \label{P1.1a}\\
    & \sum_{m\in\mathbf{M}_1} s_m\leq 1, \tag{\ref{P1.1}{b}} \label{P1.1b}\\
    & w_m\geq 0, \forall m\in \mathbf{M}_1, \tag{\ref{P1.1}{c}} \label{P1.1c}\\
    & \sum_{m\in\mathbf{M}_1} w_m\leq 1, \tag{\ref{P1.1}{d}} \label{P1.1d}
\end{align}}

\noindent where $\mathcal{S}=\{s_m\}_{m\in\mathbf{M}_1}$, and $\mathbf{M}_1$ represents the set of UDs who offload tasks to the UAV, which is determined by the offloading decisions $\mathcal{A}$.

\begin{lemma}
\label{lem:lem2}
Problem $\textbf{P1.1}$ is convex. 
\end{lemma}
\begin{proof}
Since the constraints are linear, Lemma \ref{lem:lem2} can be proved by showing that the Hessian matrix of the objective function (\ref{P1.1}) is positive semi-definite.
\end{proof}

\begin{theorem}
\label{the:ra}
The optimal resource allocation coefficient, i.e., the solution of problem $\textbf{P1.1}$, can be given as follows:
\begin{sequation}
	\label{eq.solution}
	\begin{cases}
            s_m^{*} = \frac{\sqrt{\frac{\gamma_m\eta_mD_m}{F_u^{\mathrm{max}}}}}{\sum_{i\in \mathbf{M}_1}\sqrt{\frac{\gamma_i\eta_iD_i}{F_u^{\mathrm{max}}}}},\\
            w_m^{*} = \frac{\sqrt{\frac{\gamma_mD_m+(1-\gamma_m)P_mD_m}{Br_{m,u}}}}{\sum_{i\in \mathbf{M}_1}\sqrt{\frac{\gamma_iD_i+(1-\gamma_i)P_iD_i}{Br_{i,u}}}}.
	\end{cases}
\end{sequation}
\end{theorem}
\begin{proof}
Since problem $\textbf{P1.1}$ is convex, the above conclusion can be obtained by KKT conditions~\cite{boyd2004convex}.
\end{proof}

%
% Task Offloading
%
\subsubsection{Task Offloading}
\label{subsubsec:Task Offloading}

\par For UD $m$, let us define $U_m^{\text{loc}}$ as the utility of local computing and $U_m^{\text{ec}}$ as the utility of edge computing, which can be given as follows:

{\small \begin{align}
&U_m^{\text{loc}}=\gamma_mT_m^{\text{loc}}+(1-\gamma_m)E_m^{\text{loc}},\label{eq.u_loc}\\
&U_m^{\text{ec}}=\frac{Q_u^{\text{c}}}{V}E_{m,u}^{\text{c}}+\gamma_mT_m^{\text{ec}}+(1-\gamma_m)E_m^{\text{ec}}.\label{eq.u_ec}
\end{align}}

\noindent Therefore, we can design the utility function of UD $m$ as follows:
\begin{sequation}
   \label{eq.utility}
   U_m(\mathcal{A}) =
       \begin{cases}
           U_m^{\text{loc}},a_m=0,\\
           U_m^{\text{ec}},a_m=1.
       \end{cases} \\
\end{sequation}

\noindent According to the optimal resource allocation policy $\{\mathcal{F}^{*},\mathcal{W}^{*}\}$ and removing irrelevant constant terms, problem $ \textbf{P1}$ can be transformed into a task offloading problem as follows:

 {\small \begin{align}
    \textbf{P1.2}:\quad &\underset{\mathcal{A}}{\text{min}}\sum_{m\in \mathcal{M}}U_m(\mathcal{A}) \label{P1.2}\\
    \text{s.t.}\ \
    &(\ref{Pb})\ \text{and}\ (\ref{Pc}). \notag
\end{align}}

\par The offloading decision of UD $m$ depends not only on its own demand but also on the offloading decisions of the other UDs. Considering the competitive nature of task offloading among UDs, game theory is employed to solve the task offloading decision problem.
%\noindent where $U_m(\mathcal{A})$ can be defined as:
%\begin{equation}
%    \label{eq.utility}
%    U_m(\mathcal{A}) =
%        \begin{cases}
%            \gamma_mT_m^{\text{loc}}+(1-\gamma_m)E_m^{\text{loc}},a_m=0,\\
%            \frac{Q_u}{V}E_{m,u}^{\text{comp}}+\gamma_mT_m^{\text{ec}}+(1-\gamma_m)E_m^{\text{ec}},a_m=1.
%        \end{cases} \\
%\end{equation}
\par \textbf{(1) Game Formulation.} We first model the task offloading decision problem as a multi-UDs task offloading game (MU-TOG). Specifically, the MU-TOG can be defined as a triplet $\Gamma=\{\mathcal{M},\mathbb{A}, (U_m)_{m\in \mathcal{M}}\}$, which is detailed as follows:
\begin{itemize}
\item $\mathcal{M}=\{1,2,\dots,M\}$ denotes the set of players, i.e., all UDs.
\item $\mathbb{A}=\mathbf{A}_1\times\dots\times\mathbf{A}_M$ denotes the strategy space, wherein $\mathbf{A}_m=\{0,1\}$ is the set of offloading strategies for player $m\ (m\in \mathcal{M})$, $a_m\in\mathbf{A}_m$ denotes the offloading decision of player $m$, and $\mathcal{A}=(a_1,\dots,a_M)\in \mathbb{A}$ is the strategy profile.
\item $(U_m)_{m\in \mathcal{M}}$ is the utility function of player $m$ that maps each strategy profile $\mathcal{A}$ to a real number.
\end{itemize}
\noindent Each player aims to minimize its utility by choosing a proper offloading strategy. Mathematically, the MU-TOG can be described by the following distributed optimization problem:
\begin{sequation}
    \label{eq.task-offloading}
    \underset{a_m}{\text{min}}\ U_m(a_m,a_{-m}),\ \forall m \in \mathcal{M},
\end{sequation}

\noindent where $a_{-m}=(a_1,\dots,a_{m-1},a_{m+1},\dots,a_M)$ denotes the offloading decisions of the other players except player $m$.
\par \textbf{(2) The solution to MU-TOG.} To determine the solution of MU-TOG, we first introduce the concept of Nash equilibrium, which describes a situation where no player has any incentive to unilaterally deviate from the current strategy.
\begin{definition}
\label{def:def1}
    The strategy profile $\mathcal{A}^*=(a_1^*,\dots,a_M^*)$ is a pure-strategy Nash equilibrium of game $\Gamma$ if and only if
    \begin{sequation}
    U_m(a_m^*,a_{-m}^*)\leq U_m(a_m^{\prime},a_{-m}^*) \quad  \forall a_m^{\prime}\in \mathbf{A}_m, m \in \mathcal{M}.
    \end{sequation}
\end{definition}
\par Next, we introduce a powerful tool, known as exact potential game~\cite{potential}, to help us study the existence of Nash equilibrium and how to obtain a Nash equilibrium solution for the MU-TOG.
\begin{definition}
    \label{def:def2}
     A game is called an exact potential game if and only if a potential function $F(\mathcal{A}): \mathbb{A} \mapsto \mathbb{R}$ exists such that
    \begin{sequation}
    \label{PG-def}
    \begin{split}
        &U_m(a_m,a_{-m})-U_m(b_m,a_{-m}) \\&=F(a_m,a_{-m})-F(b_m,a_{-m}), \forall (a_m,a_{-m}),(b_m,a_{-m})\in \mathbb{A}.   
    \end{split}
    \end{sequation}
\end{definition}
\begin{definition}
    The exact potential game with finite strategy sets always has a Nash equilibrium and the finite improvement property (FIP)~\cite{potential,2016Potential}. 
    \label{def:def_FIP}
\end{definition}
\par The FIP implies that a Nash equilibrium can be obtained in a finite number of iterations by any asynchronous better response update process.
\begin{theorem}
\par The MU-TOG is an exact potential game where the potential function $F(\mathcal{A})$ can be given as
\begin{sequation}
\label{eq.PF}
\begin{aligned}
  F(\mathcal{A})=&\sum_{i\in \mathcal{M}}a_i\left(\frac{Q_u^\mathrm{c}}{V}E_{i,u}^{\mathrm{c}}+\beta_i\sum_{j\leq i}a_j\beta_j+\phi_i\sum_{j\leq i}a_j\phi_j\right) \\
  &+\sum_{i\in \mathcal{M}}(1-a_i)U_i^{\mathrm{loc}},\ \forall j\in \mathcal{M},
\end{aligned}
\end{sequation}
\noindent where $\beta_i=\sqrt{\frac{\gamma_i\eta_iD_i}{F_u^{\mathrm{max}}}}$ and $\phi_i=\sqrt{\frac{\gamma_iD_i+(1-\gamma_i)P_iD_i}{Br_{i,u}}}$.
\label{theorem-PG}
\end{theorem}
\begin{proof}
The proof can refer to Theorem 3 in~\cite{JosiloD19}.
\end{proof}
\vspace{-2pt}
\par Then, let us consider the effect of constraint (\ref{Pc}) on the game. We can infer that imposing the constraint may render some strategy profiles infeasible. Suppose $\mathbb{A}^{\prime}$ is the feasible strategy space, this leads to a new game $\Gamma^{\prime}=\{\mathcal{M},\mathbb{A}^{\prime}, (U_m)_{m\in \mathcal{M}}\}$.
\begin{theorem}
\label{theo:theo_cons}
$\Gamma^{\prime}$ is also an exact potential game and has the same potential function as $\Gamma$.
\end{theorem}
\begin{proof}
The proof can refer to Theorem 2.23 in~\cite{2016Potential}.
\end{proof}
\vspace{-2pt}
\par The key idea of the MU-TOG is to utilize the FIP to update the offloading strategies of the players iteratively until the Nash equilibrium is reached, which is shown in Algorithm \ref{Algorithm 1}. The main steps of implementing the MU-TOG are described as follows. \textbf{i)} All UDs choose local computing for the initial setting (Line 1). \textbf{ii)} Each iteration is divided into $N$ decision slots (Lines 4-15). At each decision slot, one UD is selected to update its offloading decision while the offloading decisions of the other UDs remain unchanged (Line 5). \textbf{iii)} If lower utility is achieved and constraint (\ref{Pc}) is satisfied, the UD's offloading decision is updated; otherwise, the original offloading decision is maintained (Lines 6-14). \textbf{iv)} When no UD changes its offloading decision, the MU-TOG reaches the Nash equilibrium. 
\vspace{-8pt}
\begin{algorithm}	
    \label{Algorithm 1}
    \SetAlgoLined
    \KwIn{The UD information $\{\mathbf{St}_m^{\mathrm{UD}}(t)\}_{m\in \mathcal{M}}$ and the current UAV location $\mathbf{P}_u$.}
    \KwOut{The optimal task offloading and resource allocation decisions $\{\mathcal{A}^{*},\mathcal{F}^{*},\mathcal{W}^{*}\}$.}
    \textbf{ Initialization:} 
    The iteration number $l=1$, $\mathcal{A}^0=\emptyset$ and $\mathcal{A}^1=\{0,\dots,0\}$\;
    \Repeat{$\mathcal{A}^{l-1} = \mathcal{A}^l$}
    {
        $\mathcal{A}^{l-1}=\mathcal{A}^{l}$\;
        \For{\text{UD} $m\in \mathcal{M}$}
        {
            $\mathbf{A}^{l}(m)=a_m^{\mathrm{ec}}=1$\;
            Obtain $F_m^{*}$ and $w_m^{*}$ based on Eq. \eqref{eq.solution}\;
            Calculate $T_m^{\mathrm{ec}}$ based on Eq. (\ref{eq.ec-delay})\;
            Calculate $U_m^{\mathrm{ec}}$ based on Eq. (\ref{eq.u_ec})\;
            \If{$T_m^{\mathrm{ec}}\geq T_m^{\mathrm{max}}$}
            {
                $\mathbf{A}^{l}(m)=a_m^{\mathrm{loc}}=0$\;
            }
            \If{$U_m^{\mathrm{ec}}\leq U_m^{\mathrm{loc}}$}
            {
                $\mathbf{A}^{l}(m)=a_m^{\mathrm{loc}}=0$\;
            }            
        }
        Update $l=l+1$\;
    }
    $\mathcal{A}^{*}=\mathcal{A}^{l}$\;
    Obtain $\{\mathcal{F}^{*},\mathcal{W}^{*}\}$ based on Eq. \eqref{eq.solution}\;
    \Return{$\{\mathcal{A}^{*},\mathcal{F}^{*},\mathcal{W}^{*}\}$.}
    \caption{The First Stage Algorithm}
    \vspace{-2pt}
\end{algorithm}
%
% UAV Movement  
%
\subsection{Stage 2: UAV Trajectory Planning}
\label{subsec:UAV Movement}
\par Given the optimal task offloading decisions $\mathcal{A}^{*}$ and resource allocation $\{\mathcal{F}^{*},\mathcal{W}^{*}\}$, while removing irrelevant constant terms, 
problem $\textbf{P}^{\prime}$ can be converted into the subproblem $\mathbf{P2}$ to decide the UAV trajectory planning, which is expressed as follows:

{\small \begin{align}
    &\mathbf{P2}:\underset{\mathbf{P}_{u^\prime}}{\text{min}}\ V \sum_{m\in \mathbf{M}_1}\frac{\gamma_mD_m+(1-\gamma_m)P_mD_m}{w_m^{*}B\log_2(1+\frac{\phi_m}{(\|\mathbf{P}_{u^\prime}-\mathbf{P}_m\|^2+H^2)^{\mu}})}+\notag \\
    &Q_u^{\text{p}}\left(C_1\left(1+\frac{3 v_u^2}{U_{\text {p}}^2}\right)+C_2\sqrt{\sqrt{C_3+\frac{v_u^4}{4}}-\frac{v_u^2}{2}}+C_4v_u^3\right)\tau \label{P2}\\
    &\text{s.t.}\ (\ref{Ph})-(\ref{Pi}) \notag
\end{align}}

\noindent where $v_u=\frac{\|\mathbf{P}_{u^\prime}-\mathbf{P}_{u}\|}{\tau}$. Obviously, the objective function (\ref{P2}) is non-convex with respect to $\mathbf{P}_{u^\prime}$ due to the non-convex terms $TM_0=C_2\sqrt{\sqrt{C_3+\frac{v_u^4}{4}}-\frac{v_u^2}{2}}$ and $\{TM_m=\frac{1}{\log_2\left(1+\frac{\phi_m}{(\|\mathbf{P}_{u^\prime}-\mathbf{P}_m\|^2+H^2)^{\mu}}\right)}\}_{m\in \mathbf{M}_1}$. Therefore, it is difficult to directly solve problem $\textbf{P}^{\prime}$. We next transform the objective function into a convex function by introducing slack variables.
\par For the non-convex term $TM_0$, we introduce the slack variable $y$ such that $y=TM_0$ and add the following constraint:
\begin{sequation}
\label{eq.slack1}
    y \geq \sqrt{\sqrt{C_3+\frac{v_u^4}{4}}-\frac{v_u^2}{2}} \Longrightarrow \frac{C_3}{y^2} \leq y^2+v_u^2.
\end{sequation}
\par For the non-convex term $TM_m$, we introduce the slack variable $z_m$ such that $z_m=TM_m$ and add the following constraint:
\begin{sequation}
\label{eq.slack2}
    z_m \leq \log _2\left(1+\frac{\phi_m}{\left(H^2+\left\|\mathbf{P}_{u^{\prime}}-\mathbf{P}_m\right\|^2\right)^\mu}\right).
\end{sequation}
\par According to the above-mentioned relaxation transformation, problem $\textbf{P2}$ can be equivalently transformed as follows:
{\small \begin{align}
    \mathbf{P2}^{\prime}:&\ \underset{\mathbf{P}_{u^\prime},y,z_m}{\text{min}}\ Q_u\left(P_0\left(1+\frac{3 v_u^2}{U_{\text {tip }}^2}\right)+C_2y+C_3v_u^3\right)\tau\notag\\
    &+V \sum_{m\in \mathbf{M}_1}\frac{\gamma_mD_m+(1-\gamma_m)P_mD_m}{w_m^{*}Bz_m}\label{P2_1}\\
    \text{s.t.} \ &(\ref{Pi}),(\ref{eq.slack1})\ and\ (\ref{eq.slack2}) \notag
\end{align}}
\begin{theorem}
\label{the:the2}
Problem $\mathbf{P2}^{\prime}$ is equivalent to problem $\mathbf{P2}$.
\end{theorem}
\begin{proof}
\label{pro:equ}
Suppose $\{\mathbf{P}^{*}_{u^\prime},y^{*},z^{*}_m\}$ is the optimal solution of problem $\mathbf{P2}^{\prime}$. The following equation holds:
\begin{sequation}
\begin{aligned}
&y^{*} = \sqrt{\sqrt{C_3+\frac{(v^{*}_u)^4}{4}}-\frac{(v^{*}_u)^2}{2}},\\
&z^{*}_m = \log _2\left(1+\frac{\phi_m}{\left(H^2+\left\|\mathbf{P}^{*}_{u^{\prime}}-\mathbf{P}_m\right\|^2\right)^\mu}\right),
\end{aligned}
\end{sequation}

\noindent where $v^{*}_u=\frac{\|\mathbf{P}^{*}_{u^\prime}-\mathbf{P}_{u}\|}{\tau}$. Otherwise, we can further reduce the objective function by choosing a smaller $y$ or a larger $z_m$ without violating the constraints (\ref{eq.slack1}) and (\ref{eq.slack2}). Therefore, $\mathbf{P}^{*}_{u^\prime}$ is also the optimal solution to problem $\mathbf{P2}$. 
\end{proof}
\par For problem $\mathbf{P2}^{\prime}$, the optimization objective (\ref{P2_1}) is convex but the additional constraints (\ref{eq.slack1}) and (\ref{eq.slack2}) are still non-convex. Similar to~\cite{Yang2022,UAV-H,Ji2021}, the successive convex approximation (SCA) method is adopted to solve the non-convexity of (\ref{eq.slack1}) and (\ref{eq.slack2}).
\begin{theorem}
\label{pro:pro5-2-1}
Let $f(\mathbf{P}_{u^\prime},y)=y^2+v_u^2$ and given a local point $\mathbf{P}^{(l)}_{u^\prime}$ at the $l$-th iteration, we can obtain a global concave lower bound for $f(\mathbf{P}_{u^\prime},y)$ as
{\small \begin{align}
    \label{eq.pro5-2-1}
    f^{(l)}(\mathbf{P}_{u^\prime},y) \triangleq&\left(y^{(l)}\right)^2+2 y^{(l)}\left(y-y^{(l)}\right)+\frac{\|\mathbf{p}_{u^{\prime}}^{(l)}-\mathbf{p}_u\|^2}{\tau^2}\notag\\
    &+\frac{2}{\tau^2}(\mathbf{p}_{u^{\prime}}^{(l)}-\mathbf{p}_u)^T\left(\mathbf{p}_{u^{\prime}}-\mathbf{p}_u\right),
\end{align}}
\noindent where $y^{(l)}$ is defined as
\begin{sequation}
    \label{eq.y-l}
    y^{(l)}=\sqrt{\sqrt{C_3+\frac{\|\mathbf{p}_{u^{\prime}}^{(l)}-\mathbf{p}_u\|^4}{4 \tau^4}}-\frac{\|\mathbf{p}_{u^{\prime}}^{(l)}-\mathbf{p}_u\|^2}{2 \tau^2}} .
\end{sequation}
\end{theorem}
\begin{proof}
Since $f(\mathbf{P}_{u^\prime},y)$ is a convex quadratic form, the first-order Taylor expansion of $f(\mathbf{P}_{u^\prime},y)$ at local point $\mathbf{P}^{(l)}_{u^\prime}$ is a global concave lower bound.
\end{proof}
\begin{theorem}
\label{pro:pro5-2-2}
Let $g_m(\mathbf{P}_{u^\prime})=\log _2\left(1+\frac{\phi_m}{\left(H^2+\left\|\mathbf{P}_{u^{\prime}}-\mathbf{P}_m\right\|^2\right)^\mu}\right)$, we can obtain a global concave lower bound for $g_m(\mathbf{P}_{u^\prime})$ as
{\small \begin{align}
    \label{eq.taylor2}
    &g_m^{(l)}(\mathbf{P}_{u^\prime}) \triangleq \log _2\left(1+\frac{\phi_m}{\left(H^2+\|\mathbf{p}_{u^{\prime}}^{(l)}-\mathbf{p}_m\|^2\right)^\mu}\right)\notag \\
    & -\frac{\mu \phi_m (\log_2 e)(\|\mathbf{p}_{u^{\prime}}-\mathbf{p}_m\|^2-\|\mathbf{p}_{u^{\prime}}^{(l)}-\mathbf{p}_m\|^2)}{[\phi_m+(H^2+\|\mathbf{p}_{u^{\prime}}^{(l)}-\mathbf{p}_m\|^2)^\mu](H^2+\|\mathbf{p}_{u^{\prime}}^{(l)}-\mathbf{p}_m\|^2)}.
\end{align}}
\end{theorem}
\begin{proof}
The proof can refer to Proposition 1 in~\cite{Yang2022}.
\end{proof}
\par According to Theorems \ref{pro:pro5-2-1} and \ref{pro:pro5-2-2}, at the $l$-th iteration, constraints (\ref{eq.slack1}) and (\ref{eq.slack2}) can be approximated as:

{\small \begin{align}
    \label{eq.cons1}
    &\frac{C_3}{y^2}\leq f^{(l)}(\mathbf{P}_{u^\prime},y),\\
    &z_m\leq g_m^{(l)}(\mathbf{P}_{u^\prime}),
\end{align}}

\noindent which are convex. Therefore, problem $\mathbf{P2}^{\prime}$ is  converted into a convex optimization problem, which can be efficiently resolved by off-the-shelf optimization tools such as CVX~\cite{cvx}. We summarize the second stage algorithm in Algorithm \ref{Algorithm 2}. 
\begin{algorithm}
    \label{Algorithm 2}
    \SetAlgoLined
    \KwIn{The optimal task offloading and resource allocation decisions $\{\mathcal{A}^{*},\mathcal{F}^{*},\mathcal{W}^{*}\}$.}
    \KwOut{The next location $\mathbf{P}_{u^\prime}$.}
    \textbf{ Initialization:}
        The accuracy threshold $\varepsilon = 0.01$, the local point $\mathbf{P}^{(0)}_{u^\prime}=\mathbf{P}_{u}$, the iterative number $l=1$ and the objective function value $G^{(0)}=0$\;
    \Repeat{$|G^{(l)}-G^{(l-1)}|<\varepsilon$}{
        Calculate $y^{(l)}$ based on Eq.~\eqref{eq.y-l}\;
        Obtain the optimal position $\mathbf{P}^{*}_{u^\prime}$ and the objective value $G^{(l)}$ by solving problem $\mathbf{P2}^{\prime}$\;
        Update the local point $\mathbf{P}^{(l)}_{u^\prime}=\mathbf{P}^{*}_{u^\prime}$\;
        Update $l=l+1$\;}
    \Return{$\mathbf{P}^{*}_{u^\prime}$.}
    \vspace{-1pt}
    \caption{The Second Stage Algorithm}
\end{algorithm}
%
%Performance Analysis
%
\subsection{Main Steps of OJOA and Performance Analysis}
\label{sec:Performance Analysis}
\par In this section, the main steps of OJOA are described in Algorithm \ref{Algorithm 3}, and the corresponding analysis is given. 
\begin{theorem}
    \label{the:performance}
Assume that the proposed algorithm produces an optimality gap $C\geq 0$ in solving $\mathbf{P}^{\prime}$ and $C_s^{\mathrm{opt}}$ denotes the optimal time-average UD cost that problem $\mathbf{P}$ can achieve over all policies given full knowledge of the future computing demands and locations for all UDs, the time-average UD cost achieved by the proposed algorithm is bounded by
\begin{sequation}
    \label{eq.performance}
    \frac{1}{T} \sum_{t=1}^T\sum_{m=1}^M C_m(t) \leq C_s^{\mathrm{opt}}+\frac{WT+C }{V},
\end{sequation}

\noindent where $W$ is defined in Theorem~\ref{the:drift-plus-penalty}.
\end{theorem}
\begin{proof}
According to Lemma 4.11 in~\cite{2010Neely}, the \textit{T-slot drift-plus-penalty} achieved by the proposed algorithm ensures that
\begin{sequation}
\label{eq.opti_cost}
    L(\mathbf{Q}_u(T))-L(\mathbf{Q}_u(1))+V\sum_{t=1}^TC_s(t)\leq WT^2+CT+VTC_s^{\mathrm{opt}}.
\end{sequation}

\noindent Using the fact that $L(\mathbf{Q}_u(T))\geq 0$ and $L(\mathbf{Q}_u(1))=0$, and dividing by $VT$ for the above inequality, we can prove the theorem.
\end{proof}
\begin{theorem}
The proposed algorithm can satisfy the UAV energy consumption constraint defined in (\ref{eq.eng-cons}).
\end{theorem}
\begin{proof}
The proof can refer to Theorem 2 in~\cite{JiangDXI23}. 
\end{proof}
\begin{theorem}
\label{the:complexity}
The proposed OJOA has a polynomial worst-case complexity in each time slot, i.e., $\mathcal{O}(I_cM+M^{3.5}\log_2(\frac{1}{\varepsilon}))$, where $I_c$ represents the number of iterations required for Algorithm \ref{Algorithm 1} to converge to the Nash equilibrium, $M$ denotes the number of UDs and $\varepsilon$ is the accuracy of SCA for solving problem $\mathbf{P2}^{\prime}$.
\end{theorem}
\begin{proof}
OJOA contains two phases in each time slot, i.e., Algorithm \ref{Algorithm 1} and Algorithm \ref{Algorithm 2}. In Algorithm \ref{Algorithm 1}, assuming that the outer iteration (i.e., Lines $2-17$) converges after $I_c$ iterations, the computational complexity of the algorithm can be calculated as $\mathcal{O}(I_cM)$. In Algorithm \ref{Algorithm 2}, according to the analysis in~\cite{Wang2022}, the computational complexity is $\mathcal{O}(M^{3.5}\log_2(\frac{1}{\varepsilon}))$. Therefore, the computational complexity of OJOA is $\mathcal{O}(I_cM+M^{3.5}\log_2(\frac{1}{\varepsilon}))$ in the worst case.
\end{proof}

\par {\color{b}Accordingly, it is proven that the proposed algorithm can effectively guarantee the performance of the system, meet the UAV energy consumption constraint and have low computational complexity.}
\begin{algorithm}
    \label{Algorithm 3}
    \SetAlgoLined
    \KwIn{The energy queue $Q_u^{\text{c}}(1)=0$, $Q_u^{\text{p}}(1)=0$ and the control parameter $V$.}
    \KwOut{time-average UD cost $TSC$.}
    \textbf{ Initialization:} 
    Initialize $TSC = 0$ and the initial position of the UAV $\mathbf{P}_u(1)=\mathbf{P}_I$\;
    \For{$t=1$ to $t=T$}
    {
        Acquire the UD information $\{\mathbf{St}_m^{\text{UD}}(t)\}_{m\in \mathcal{M}}$\; 
        With fixed $\mathbf{P}_u(t)$, call Algorithm \ref{Algorithm 1} to obtain $\{\mathcal{A}^{*},\mathcal{F}^{*},\mathcal{W}^{*}\}$\;
        With fixed $\{\mathcal{A}^{*},\mathcal{F}^{*},\mathcal{W}^{*}\}$, call Algorithm \ref{Algorithm 2} to obtain $\mathbf{P}_{u^\prime}^{*}$\;
        All UDs perform their tasks based on $\mathcal{A}^{*}$ and obtain corresponding cost $C_m^{*}(t)$\;
        The UAV provides MEC service to the UDs and flies towards position $\mathbf{P}^{*}_{u^\prime}$\;
        System cost $C_s(t)=\sum_{m=1}^MC_m^{*}(t)$\;
        $TSC=TSC+C_s(t)$\;
        Update the energy queue $\mathbf{Q}_u(t+1)$ according to Eq.~\eqref{eq.queue}\;
        Update $t=t+1$\;
    }
    $TSC=TSC/T$\;
    \Return{$TSC$.}
    \caption{OJOA}
\end{algorithm}
%
% Simulation Results and Analysis
%
\section{Simulation Results}
\label{sec:Simulation Results and Analysis}

\par In this section, we perform simulations to validate the effectiveness of our proposed OJOA.
%
% Simulation Setups
%
\subsection{Simulation Setup}
\label{subsec:Simulation setups}
\begin{figure*}[!hbt] 
	\centering
	\setlength{\abovecaptionskip}{1pt}%    
	\setlength{\belowcaptionskip}{1pt}%
	\subfigure[Time-average UD cost]
	{
		\begin{minipage}[t]{0.31\linewidth}
			\centering
			\includegraphics[scale=0.4]{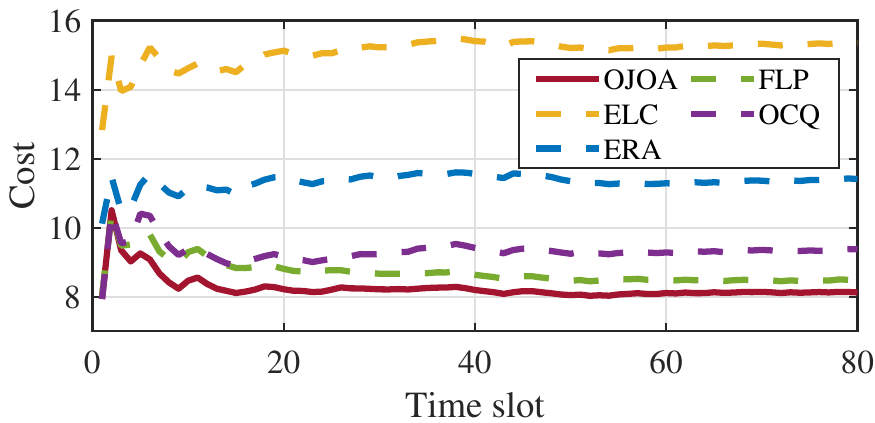}
		\end{minipage}
	}
	\subfigure[Time-average UAV energy consumption]
	{
		\begin{minipage}[t]{0.31\linewidth}
			\centering
			\includegraphics[scale=0.4]{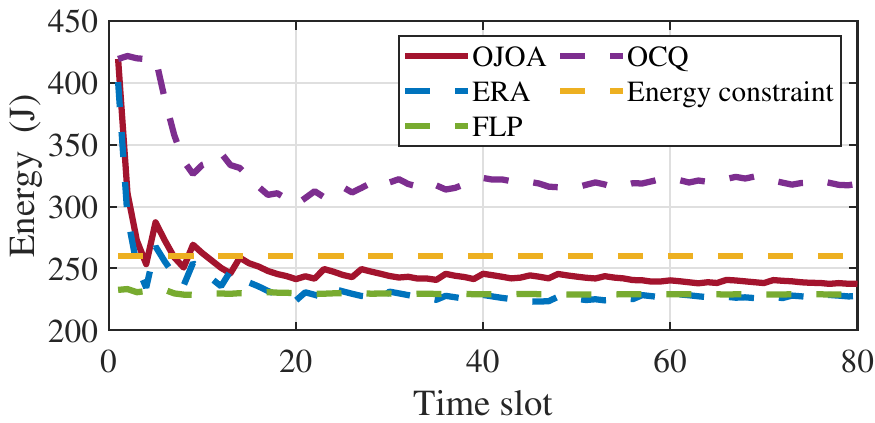}	
		\end{minipage}
	}
	\subfigure[Time-average UAV workload]
	{
		\begin{minipage}[t]{0.31\linewidth}
			\centering
			\includegraphics[scale=0.4]{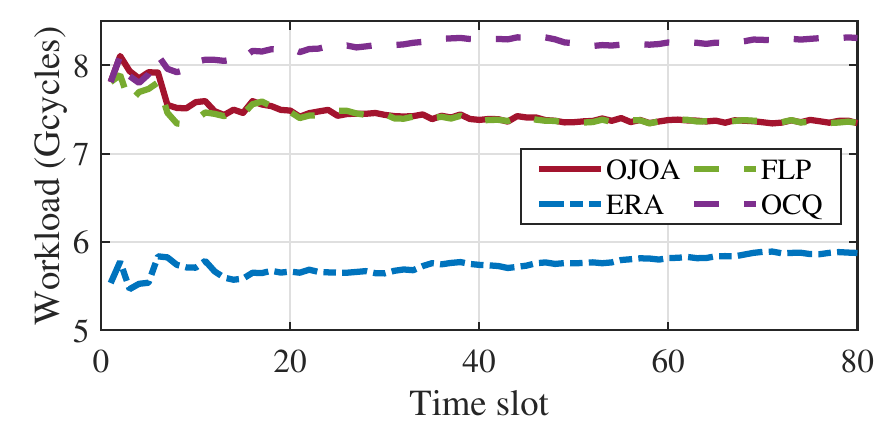}
		\end{minipage}
	}
	\caption{System performance with respect to the time slots. (a) Time-average UD cost. (b) Time-average UAV energy consumption. (c) Time-average UAV workload.}
	\label{fig_time}
	\vspace{0pt}
\end{figure*}
\begin{figure*}[!hbt] 
	\centering
	\setlength{\abovecaptionskip}{2pt}%    
	\setlength{\belowcaptionskip}{2pt}%
	\subfigure[Time-average UD cost]
	{
		\begin{minipage}[t]{0.31\linewidth}
			\centering
			\includegraphics[scale=0.40]{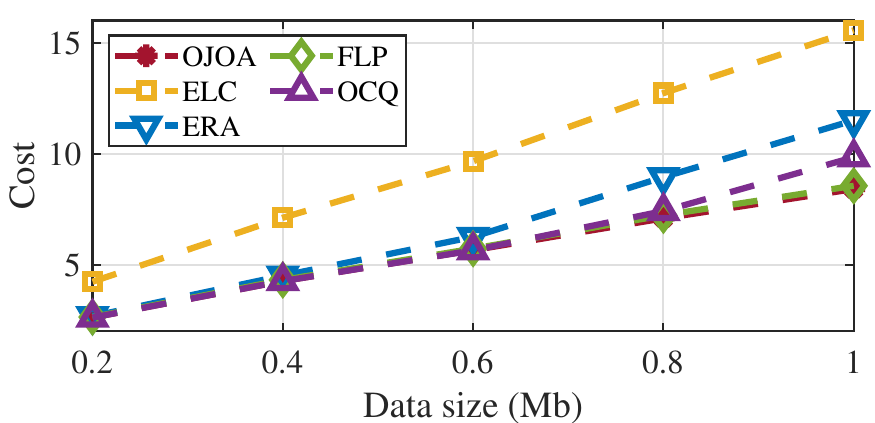}
		\end{minipage}
	}
	\subfigure[Time-average UAV energy consumption]
	{
		\begin{minipage}[t]{0.31\linewidth}
			\centering
			\includegraphics[scale=0.40]{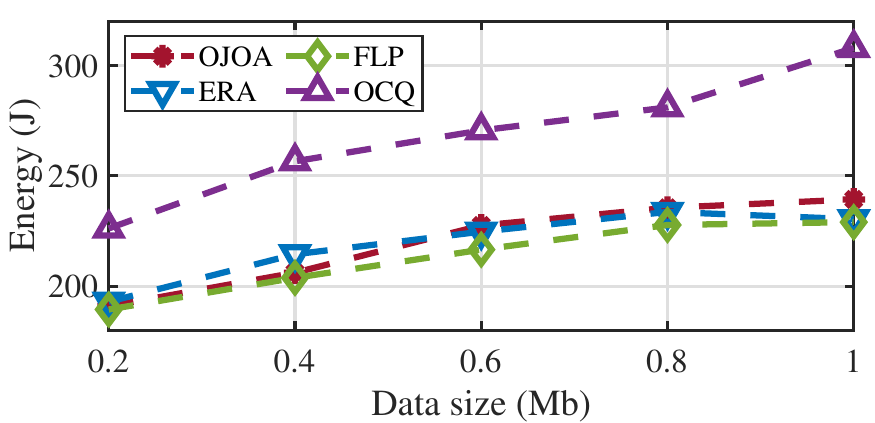}	
		\end{minipage}
	}
	\subfigure[Time-average UAV workload]
	{
		\begin{minipage}[t]{0.31\linewidth}
			\centering
			\includegraphics[scale=0.40]{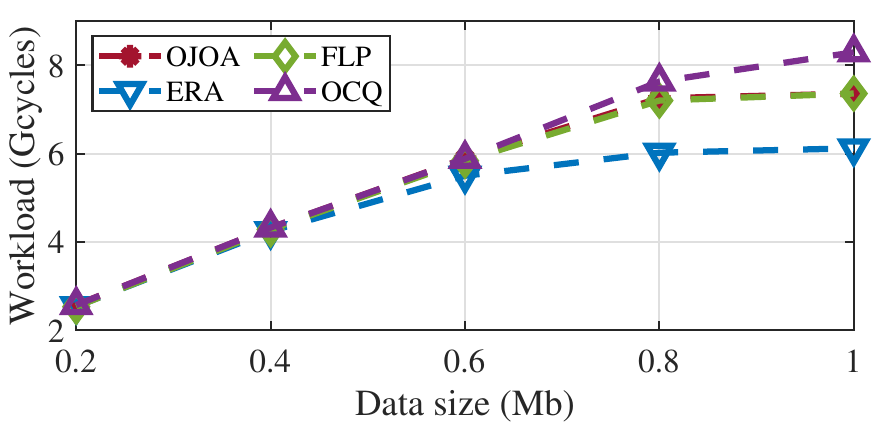}
		\end{minipage}
	}
	\caption{System performance with respect to the task data size. (a) Time-average UD cost. (b) Time-average UAV energy consumption. (c) Time-average UAV workload.}
	\label{fig_dataszie}
	\vspace{-1em}
\end{figure*}
\par We consider a UAV-enabled MEC system consisting of a UAV and $20$ UDs, where the initial horizontal position of the UAV is set as $\mathbf{P}_{I}=[200,200]$, the fixed height is $H=100\ \text{m}$, and the initial positions of UDs are distributed in the area of $400 \times 400\ \text{m}^2$. The system timeline is discretized into $80$ time slots and the length of each time slot is $1\ \text{s}$~\cite{Wang2022}. The maximum speed of the UAV is set to $v_u^{\text{max}}=30\ \text{m/s}$~\cite{Yang2022} and the total computing resources of the UAV are defined as $  F_{u}^{\text{max}}=20\ \text{GHz}$. The computing capacity of UDs is randomly taken from $\{1,1.5,2\}\  \text{GHz}$, and the transmit power is set to $P_m=0.1\ \text{W}$. Each UD generates a computing task per time slot with input data size $D_m(t)\in [0.1, 1]\ \text{Mb}$, computation intensity $\eta_m(t)\in [500, 1500]\ \text{cycles/bit}$~\cite{Sun2023}, and maximum tolerable delay $T_m^\text{max}=1\ \text{s}$~\cite{Wang2022}. The channel bandwidth is set to $B=4\ \text{MHz}$. Moreover, we compare OJOA with the following four benchmark schemes:
\begin{itemize}
    \item Entire local computing (ELC): All UDs process their tasks locally.
    \item Equal resource allocation (ERA)~\cite{Josilo}: The UAV allocates computing and communication resources equally. 
    \item Fixed location deployment (FLP): The UAV hovers over the center of the service area to provide edge computing services.
    \item Only consider QoE (OCQ)~\cite{JiangDXI23}: Ignoring the UAV energy consumption constraint, all decisions are made only to minimize the time-average UD cost.
\end{itemize}
%
% Evaluation results
%

\subsection{Evaluation Results}
\label{subsec:Evaluation results}
\par \textbf{\textit{Impact of Time.}} Figs.~\ref{fig_time}(a), \ref{fig_time}(b), and \ref{fig_time}(c) show the dynamics of time-average UD cost, time-average UAV energy consumption, and time-average UAV workload among the five schemes. First, ELC exhibits the worst performance for time-average UD cost.  Obviously, this is because all tasks are executed locally on UDs. Furthermore, ERA shows poorer performance in terms of time-average UD cost compared to FLP, OCQ, and OJOA. The reason is that due to UDs' heterogeneous computing requirements, the average resource allocation strategy cannot effectively utilize the limited computing and communication resources. It also explains that ERA has the lowest time-average UAV workload and time-average UAV energy consumption. In addition, it can be observed that OCQ achieves higher time-average UD cost compared to FLP and OJOA. This is mainly because of the game theory-based task offloading algorithm, which is detailed in Section \ref{subsubsec:Task Offloading}. Specifically, regardless of the UAV energy consumption constraint, more UDs choose to offload tasks to the UAV, which leads to a heavier UAV workload. Finally, OJOA shows superior performance in the time-average UD cost among the five schemes and satisfies the UAV energy consumption constraint. This is because OJOA optimizes the trajectory of the UAV and adopts the optimal resource allocation strategy.

\par \textbf{\textit{Impact of Data Size.}} Figs.~\ref{fig_dataszie}(a), \ref{fig_dataszie}(b) and \ref{fig_dataszie}(c) show the impact of the task data size on time-average UD cost, time-average UAV energy consumption, and time-average UAV workload among the comparative schemes, respectively. First, it can be observed that the time-average UD cost, time-average energy consumption, and time-average UAV workload show an upward trend with the increasing task data size. This is expected as the larger task data size leads to higher overheads on computing, communication, and energy consumption for UDs and the UAV. Furthermore, we can see that ERA, OCQ, and OJOA achieve similar time-average UD cost when the task data size is relatively small (less than 0.4 Mb). The reason is the UAV has enough resources to process the tasks of UDs when the data size is small. Finally, it can be observed that the proposed OJOA is able to adapt to varying task data sizes with relatively superior performances in time-average UD cost, especially in the heavy workload scenario.
\section{Conclusion}
\label{sec:Conclusion}
\par In this work, we study task offloading, resource allocation, and UAV trajectory planning in an energy-constrained UAV-enabled MEC system. A JTRTOP is formulated to maximize the QoE of all UDs while satisfying the UAV energy consumption constraint. Since the JTRTOP is future-dependent and NP-hard, we propose the OJOA to solve the problem. Specifically, the future-dependent JTRTOP is firstly transformed into the PROP by using Lyapunov optimization methods. Furthermore, a two-stage optimization algorithm is proposed to solve the PROP. Simulation results show that OJOA outperforms the conventional approaches in terms of time-average UD cost while meeting the UAV energy consumption constraint.
\section*{Acknowledgement}
\par This study is supported in part by the National Natural Science Foundation of China (62172186, 62002133, 61872158, 62272194), and in part by the Science and Technology Development Plan Project of Jilin Province (20230201087GX).
\newpage
\bibliographystyle{IEEEtran}
\bibliography{myref}
\end{document}